\documentclass[10pt,twocolumn,letterpaper]{article}

\usepackage[preprint]{cvpr}      
\usepackage{amsmath}
\usepackage{amssymb}
\usepackage{amsthm}
\usepackage{booktabs}
\usepackage{multirow}
\usepackage{graphicx}

\newtheorem{theorem}{Theorem}[section]
\newtheorem{lemma}{Lemma}[section]
\newtheorem{corollary}{Corollary}[section]
\newtheorem{proposition}{Proposition}[section]
\newtheorem{definition}{Definition}[section]

\definecolor{cvprblue}{rgb}{0.21,0.49,0.74}
\usepackage[pagebackref,breaklinks,colorlinks,allcolors=cvprblue]{hyperref}

\def\paperID{*****} 
\def\confName{CVPR}
\def\confYear{2026}

\title{Revisiting Certified Defense with Differential Privacy on Vision Transformers}

\author{Jun Yan\\
SHOU\\
{\tt\small yanjun@ieee.org}
\and
Weiquan Huang\\
Tongji University\\
{\tt\small weiquanh@tongji.edu.cn}
\and
Qixian Zhang\\
Tongji University\\
{\tt\small zhangqx@tongji.edu.cn}
\and
Yan Bai\\
Independent Researcher\\
{\tt\small baiyan1996@icloud.com}
\and
Shutai Zhang\\
SHOU\\
{\tt\small stzhang@shou.edu.cn}
}

\begin{document}
\maketitle
\begin{abstract}
Certified defenses that incorporate differential privacy have proven effective on Convolutional Neural Networks (CNNs), furnishing rigorous robustness guarantees against norm-bounded adversaries. However, the certified robustness behavior of Pixel Differential Privacy (PixelDP) remains largely unexplored with the self-attention architecture now dominating the deep-learning landscape. Given that the Transformer has a profound impact on our daily applications from the digital world to the physical world, it is crucial to study certified robustness through differential-privacy-style stability. To fill this research gap, we revisit this construction in Vision Transformers and identify a failure mode that is largely hidden in the convolutional setting. When noise is injected after the patch embedding, the Laplace mechanism with the inherited grouped $\ell_1$ sensitivity bound collapses to chance-level accuracy across noise scales, whereas the Gaussian mechanism remains trainable. This contrast isolates the source of failure: not the injected noise itself, but the geometry of the sensitivity constraint. We show that the attenuation induced by the inherited $\Delta_{1,1}$ projection increases with layer width and kernel size according to a random-matrix scale $C/(\sqrt{M}+\sqrt{N})$. Replacing the $\ell_1$-type constraint with a spectral-norm constraint eliminates the collapse across datasets and architectures, but creates a fundamental obstacle: the repaired models no longer satisfy the sensitivity condition required by the standard Laplace certificate. We resolve this mismatch by deriving a dimension-free $(\varepsilon,\ \delta)$-privacy guarantee for the Laplace mechanism under $\ell_2$ sensitivity through concentration of the privacy loss. The resulting certificate applies directly to the repaired models without retraining and avoids the $\sqrt{d}$ privacy degradation incurred by the naive conversion from $\ell_2$ to $\ell_1$ sensitivity. Experiments on CIFAR-10, CIFAR-100, and SVHN reproduce the failure-and-repair pattern, while ImageNet confirms that the inherited grouped-bound failure persists at scale. Finally, matched-budget experiments show that adversarial training improves certified accuracy at the evaluated positive radii.

\end{abstract}
    
\section{Introduction}
\label{sec:intro}

While deep neural networks (DNNs) empower social life, they also raise various security and privacy concerns~\cite{bengio2024managing}. During inference, a network may be threatened by adversarial examples. An attacker can cause the model to make incorrect predictions by adding imperceptible perturbations~\cite{goodfellow2015explaining,madry2018towards}. It poses serious risks to safety-critical applications such as autonomous driving~\cite{eykholt2018robust} and robotics~\cite{wang2025exploring}.
\par In addition to empirical defenses~\cite{madry2018towards, zhang2019theoretically,yan2024enhance}, the certified defense approaches~\cite{cohen2019certified, zhang2020towards,li2023sok} offer provable robustness guarantees under explicit assumptions. PixelDP~\cite{lecuyer2019certified} is a general framework that leverages differential-privacy-style stability~\cite{dwork2006differential,abadi2016deep} to certify robustness against norm-bounded input perturbations. It constrains the sensitivity of a pre-noise mapping and injects calibrated noise
to stabilize the randomized prediction.
\par However, previous research on the certified robustness of differential privacy mainly focuses on  Convolutional Neural Networks (CNNs)~\cite{szegedy2016rethinking, he2016deep}. Transformers use self-attention to model global interactions between tokens~\cite{vaswani2017attention,dosovitskiy2021image}. The recent scaling law~\cite{kaplan2020scaling} suggests that increasing model and data scale can improve predictive performance~\cite{kaplan2020scaling}. Transformers can be trained on huge amounts of data and are hardware friendly for large-scale parallel processing, making them widely used as the backbone for building general artificial intelligence base models~\cite{bubeck2023sparks}. The large language model serving the public in the cloud~\cite{guo2025deepseek} and the general-purpose robot operating at the edge~\cite{kim2025openvla} are built on this transformer architecture that can predict the next token. Nevertheless, it is still insufficient to explore the robustness of defense with differential privacy in this new architecture.
\par To bridge this gap, we rethink the certified defense with differential privacy on the vision transformer architectures. Our study aims to resolve such a research problem:\\
\indent \textit{How can PixelDP be ported to self-attention architectures while preserving a valid and useful robustness certificate?}
\par This paper revisits the pipeline in the new setting, as illustrated in Fig.~\ref{fig:PixelDPVIT_Framework}. Addressing this problem requires resolving two technical challenges. First, standard dot-product self-attention is not globally Lipschitz on an unbounded domain. We therefore inject noise immediately after the patch embedding, a linear layer whose sensitivity can be explicitly controlled, exactly as PixelDP treats a first convolution. Our experiments confirm the distinction sharply. Noise placed after an attention block is absorbed by training and leaves accuracy at the baseline for every noise level, which certifies nothing, whereas noise after the patch embedding behaves as the theory prescribes. Second, PixelDP's Gaussian arm degrades gracefully under the faithful port, while the Laplace arm may collapse to chance accuracy even at the smallest noise level. The failure is not a matter of noise magnitude, since the Gaussian mechanism injects a larger standard deviation at the same attack bound and trains without difficulty. Nor is it intrinsic to the Laplace distribution; the original PixelDP pipeline trains Laplace models on the same grid without incident, on convolutions. The culprit is the grouped convolution bound inherited by the Laplace route. For a non-overlapping patch embedding, this bound aggregates all spatial kernel locations of an input channel and is substantially more conservative than the exact stride-aware $\ell_1\!\to\!\ell_1$ operator norm. We prove a random-matrix law showing that, relative to the spectral projection used by the Gaussian route, the resulting attenuation scales as $C/(\sqrt{M}+\sqrt{N})$ under the Gaussian initialization model. At the measured ViT patch embedding, this large norm ratio predicts a severe drop in the local signal-to-noise ratio, which makes optimization substantially harder and accounts for the failure we observe across all training recipes.
\begin{figure*}[!t]
\centering
\includegraphics[width=0.9\textwidth]{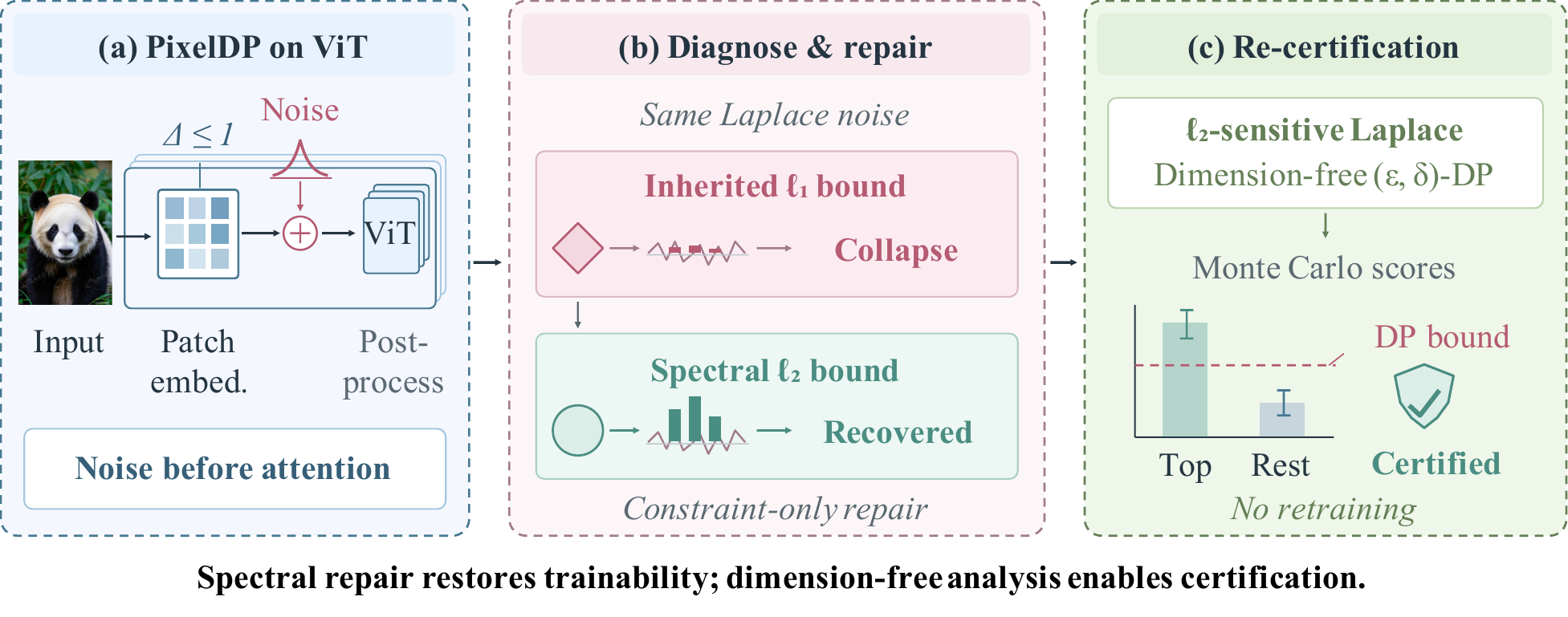}
\caption{The framework of scaling PixelDP on ViT.}
\label{fig:PixelDPVIT_Framework}
\end{figure*}
This diagnosis suggests its own experiment. We decouple the constraint from the mechanism and train Laplace models under the spectral constraint of the Gaussian route, changing nothing else. The spectral repair also restores trainability across CIFAR-10, CIFAR-100, and SVHN. The grouped-bound failure also persists at ImageNet scale on the autoencoder configuration of the original pipeline. The causal chain is complete, but the repair is not yet a defense. The Laplace mechanism's $(\varepsilon,0)$ guarantee is proved from an $\ell_1$ sensitivity bound, the spectral constraint provides an $\ell_2$ bound, and the naive conversion between the two costs a factor of $\sqrt{d}$ in the privacy budget, about $111$ at our noise layer, which would demand a hundredfold larger noise.

To close this certification gap in PixelDP, we specialize the existing $\ell_2$-sensitivity Laplace guarantee to our setting and give a direct privacy-loss derivation for completeness. The analysis framework needs only a closed form for the shifted first absolute moment of a Laplace variable and Hoeffding's inequality applied to the privacy loss, whose coordinates are independent and bounded. The guarantee plugs directly into the PixelDP certification procedure, holds at every attack radius without group privacy, and applies to our trained checkpoints as a re-certification, with no retraining. It also quantifies an honest cost. At matched $(\varepsilon,\delta)$ the calibrated Laplace noise carries about $1.47\times$ the standard deviation of the Gaussian mechanism, so Gaussian remains the efficient choice for pure $\ell_2$ defense, and the value of the theorem is that it closes a gap the original framework left open and turns a diagnostic into a certified configuration.

Finally we ask whether the certified accuracy of the faithful port can be improved rather than merely explained. Following the margin perspective of adversarially trained smoothing~\cite{salman2019provably}, we combine the noise layer with $\ell_2$ adversarial training under a matched-budget protocol that isolates the effect of the attack from the training recipe. On all three datasets, adversarial training improves certified accuracy over the nonzero operating region of the matched controls. On SVHN the joint training collapses to the class prior and a two-phase schedule, noisy pretraining followed by adversarial fine-tuning, restores it, which we report as a practical recipe.

Our contributions are the following.
\begin{itemize}
\item A faithful port of PixelDP to Vision Transformers with a placement analysis showing that the patch embedding provides a simple and directly enforceable pixel-space sensitivity bound.
\item We identify severe over-attenuation under the inherited grouped $\ell_1$ bound. A random-matrix law predicts its scale, while exact stride-aware and spectral controls isolate the effect of the sensitivity constraint.
\item We connect the repaired spectral configuration to an existing approximate-DP guarantee for coordinatewise Laplace noise under $\ell_2$ sensitivity~\cite{andersson2024count}, derive the corresponding closed-form PixelDP calibration, and show that existing repaired checkpoints can be certified without retraining.
\item An ImageNet study across Inception and ViT-Base backbones confirming the post-processing property empirically, and a matched-budget adversarial training ablation showing certified gains at every radius on CIFAR-10, CIFAR-100, and SVHN.
\end{itemize}

\section{Related Work}
\label{sec:related_works}
\paragraph{Certified defenses from noise.}
PixelDP~\cite{lecuyer2019certified} connects differential privacy~\cite{dwork2006differential,dwork2014algorithmic} to adversarial robustness. It injects calibrated noise after a sensitivity-constrained layer and certifies predictions through the stability of expected scores. Randomized smoothing~\cite{cohen2019certified} certifies randomized classifiers, with the standard Gaussian construction injecting noise directly at the input and admitting a tighter $\ell_2$ certificate. Our goal is therefore diagnostic rather than to outperform specialized
input-space smoothing on pure $\ell_2$ certification. SmoothAdv~\cite{salman2019provably} shows that adversarially training the smoothed classifier enlarges its margins and improves certified accuracy, and later work extends smoothing to further settings~\cite{alfarra2022data,rekavandi2024certified,ugare2024incremental}. Gaussian differential privacy~\cite{dong2022gaussian} offers an alternative accounting language that we use for a secondary calibration. Our study differs from this line in its object. We do not propose a new smoothing distribution. Instead, we study how PixelDP's inherited mechanism--constraint pairings interact with ViT patch embeddings and how the resulting spectral repair can be certified under $\ell_2$ sensitivity.

\paragraph{Robustness of Vision Transformers.}
The Lipschitz behavior of attention has been studied directly. Self-attention is not Lipschitz on unbounded domains~\cite{kim2021lipschitz}, variants restore control through architectural change~\cite{qi2023lipsformer}, and sharp bounded-domain estimates exist for normalized token sets~\cite{castin2024smooth,khromov2024some}. These results concern token-space smoothness of the attention map.  They do not directly yield the simple, enforceable pixel-space sensitivity bound required by PixelDP: an $\ell_p$ pixel perturbation generally affects multiple tokens, while existing attention-side bounds depend on domain or architectural assumptions and do not directly control the full normalization--attention prefix from pixels to tokens. We therefore place the noise before the first attention block, immediately after the patch embedding, whose sensitivity can be explicitly constrained. On the empirical side, adversarial training recipes for ViTs have matured~\cite{mo2022adversarial,debenedetti2023light}, and we adopt the warmup and weak augmentation guidance of the light recipe~\cite{debenedetti2023light} for our ablation, with robust overfitting handled by early stopping~\cite{rice2020overfitting} and adaptive evaluation following standard guidance~\cite{athalye2018obfuscated,croce2020reliable,carlini2017towards}.

\paragraph{Privacy and robustness together.}
Differentially private training protects the training set~\cite{abadi2016deep} and can be combined with smoothing for joint guarantees~\cite{wu2024augment,lyu2024adaptive}. PixelDP and our work use privacy as an analysis tool for inference-time robustness rather than for training data. The two goals are orthogonal and compatible.

\section{Porting PixelDP to Vision Transformers}
\label{sec:method}

\subsection{Where the Noise Can Go}
\label{sec:placement}
The construction of Sec.~\ref{sec:prelim} needs a pre-noise map $g$ whose sensitivity can be bounded and enforced during training. Standard dot-product self-attention is not globally Lipschitz on an unbounded token domain~\cite{kim2021lipschitz}. Although finite bounds can in principle be derived on bounded domains, available attention-side bounds are architecture- and domain-dependent and do not directly provide the simple, enforceable pixel-space sensitivity bound required by our PixelDP protocol. We therefore place the noise immediately after the patch embedding, whose sensitivity can be controlled explicitly.

What remains is the patch embedding, a single strided convolution that is linear in the input. We therefore inject the noise immediately after the patch embedding and enforce its sensitivity to one after every optimizer step, exactly as PixelDP treats a first convolution. This choice has a convenient consequence. Because the noise sits directly after the linear embedding whose $\Delta_{2,2}$ equals one, token-space guarantees transfer to pixel space with constant exactly one, and no pullback factor appears. The patch embedding is therefore a simple and directly auditable placement for our PixelDP port. When we instead add noise after the first attention block, training absorbs the perturbation completely, accuracy stays at the noise-free baseline for every noise level on SVHN, and nothing is certified because no sensitivity bound exists. We report this absorption experiment in the supplement as evidence that the placement analysis is not a formality.

\subsection{Sensitivity-Induced Signal Collapse}
\label{sec:collapse}
Under the faithful port a sharp asymmetry appears. The Gaussian arm degrades gracefully along the standard grid $L\in\{0.03,0.1,0.3,1\}$, while the Laplace arm collapses to chance at $L=0.03$ on every dataset (Fig.~\ref{fig:repair}). Noise magnitude does not explain this. At the same numerical radius $L$, the Gaussian $\ell_2$ route injects standard deviation $2.54L$, while the Laplace $\ell_1$ route injects $1.41L$. This comparison is a noise-amplitude diagnostic, not a like-for-like comparison of the two threat sets. The mechanism itself does not explain it either, since the original paper trains Laplace models on the same grid on convolutions. The difference is the constraint that each mechanism inherits: a conservative grouped $\ell_1$ bound for Laplace and spectral $\ell_2$ control for Gaussian.

Let $s_{\rm grp}(W)=\max_j\sum_{m,a,b}|W_{m,j,a,b}|$ denote the inherited grouped bound of Eq.~(\ref{eq:s1}). Renormalizing $s_{\rm grp}(W)$ to one aggregates $C=Mk^2$ coefficients per input-channel group. Renormalizing the spectral norm divides it by $\sigma_{\max}$, which grows only like $\sqrt{M}+\sqrt{N}$. This gap predicts severe attenuation under the inherited grouped bound.

\begin{proposition}[Norm-ratio law]
\label{prop:collapse}
\begin{equation}
\label{eq:rho_divice}
\rho
\ge
\frac{
C\sqrt{2/\pi}-\sqrt{2C\ln(1/\delta_1)}
}{
\sqrt{M}+\sqrt{N}+\sqrt{2\ln(1/\delta_2)}
}.
\end{equation}
Thus the lower bound has leading-order scale
\begin{equation}
\rho
=
\Omega\!\left(
\frac{C}{\sqrt{M}+\sqrt{N}}
\right).
\end{equation}
Motivated by the leading Gaussian scales of the numerator and
denominator, we use
\begin{equation}
\label{eq:ratio}
\rho_{\mathrm{pred}}
=
\frac{C\sqrt{2/\pi}}{\sqrt{M}+\sqrt{N}}
\end{equation}
as an initialization-scale diagnostic.
\end{proposition}
The proof, in the supplement, uses Gaussian concentration of the column mass and the Davidson--Szarek bound on the spectral norm~\cite{vershynin2018high}. The leading constant is accurate at initialization and provides a useful architecture-level scale. Evaluating $C\sqrt{2/\pi}/(\sqrt{M}+\sqrt{N})$ at the three first layers in our study predicts $12.5$, $34.8$, and $118$, and direct measurement on the same weights gives $13$, $38$, and $123$ (Fig.~\ref{fig:law}).

\begin{corollary}[Linear-component attenuation]
\label{cor:snr}
Consider the same weight matrix $W$ under the two scalar projections
\[
W_{\rm grp}
=
W/s_{\rm grp}(W),
\qquad
W_{\rm spec}
=
W/\sigma_{\max}(W).
\]
For every input $x$,
\[
W_{\rm grp}x
=
\rho_{\rm grp}^{-1}W_{\rm spec}x,
\qquad
\rho_{\rm grp}
=
\frac{s_{\rm grp}(W)}{\sigma_{\max}(W)}.
\]
Hence, against the same additive noise scale, any homogeneous signal-to-noise measure applied to the weight-dependent linear component is reduced by a factor $\rho_{\rm grp}$ under the first projection. This statement does not include additive biases or subsequent positional embeddings.
\end{corollary}
At the ViT patch embedding, $\rho\approx 123$ pushes the projected signal to about $4\%$ of the injected noise at $L=0.03$, while the spectral projection leaves it at $2.4$ times the noise. This severe attenuation makes optimization substantially harder and is consistent with the failure observed across the training recipes we tested.

The diagnosis implies a causal test. Decouple the constraint from the mechanism and train the Laplace mechanism under the spectral constraint, changing nothing else. This diagnosis predicts that replacing the inherited $\Delta_{1,1}$ constraint with the spectral constraint should substantially improve trainability. The observed recovery across
datasets supports this constraint-based explanation (Sec.~\ref{sec:exp_repair}). The repaired configuration, however, is not yet certified, because the Laplace guarantee of Sec.~\ref{sec:prelim} is proved from an $\ell_1$ bound that the spectral projection does not provide. The naive conversion $\|v\|_1\le\sqrt{d}\,\|v\|_2$ costs a factor $\sqrt{d}$ in the privacy budget, and at our noise layer $d=192\times 8\times 8$ gives $\sqrt{d}\approx 110.9$, which would demand a hundredfold larger noise. This is also, we believe, the quantitative reason the original framework never offered a 2-norm Laplace configuration.

\subsection{A Dimension-Free Certificate for Laplace under $\ell_2$ Sensitivity}
\label{sec:repair}
Our main theoretical result removes the $\sqrt{d}$ obstacle. Throughout, $M(x)=g(x)+Z$ with $Z_i\sim\operatorname{Lap}(b)$
i.i.d., and
\[
v=g(x')-g(x)
\]
denotes the output shift.

\begin{lemma}[Shifted absolute moment]
\label{lem:shift}
If $z\sim\mathrm{Lap}(0,b)$ then $\mathbb{E}\,|z-v|=|v|+b\,e^{-|v|/b}$ for every $v\in\mathbb{R}$, and consequently
\begin{equation}
\label{eq:perterm}
\mathbb{E}\bigl[\,|z-v|-|z|\,\bigr]\;=\;b\,\varphi\!\Bigl(\frac{|v|}{b}\Bigr)\;\le\;\frac{v^2}{2b},
\qquad \varphi(t)=t+e^{-t}-1 .
\end{equation}
\end{lemma}

\begin{lemma}[Tail bound implies DP]
\label{lem:tail}
Fix $x,x'$ and let $\mathrm{PL}(y)=\ln\bigl(p_{M(x)}(y)/p_{M(x')}(y)\bigr)$ be the privacy loss. If $\Pr_{y\sim M(x)}[\mathrm{PL}(y)>\varepsilon]\le\delta$, then Eq.~(\ref{eq:dp}) holds for this pair.
\end{lemma}
\par Approximate-DP guarantees for coordinatewise Laplace noise under an $\ell_2$ sensitivity bound have been established previously~\cite{andersson2024count}. We specialize this result to the PixelDP setting and give a direct privacy-loss derivation for completeness. The role of the result here is not a new Laplace privacy theorem, but to close the certification gap created by replacing the original $\Delta_{1,1}$ constraint with the spectral $\Delta_{2,2}$ constraint.
\begin{theorem}[Laplace mechanism under $\ell_2$ sensitivity]
\label{thm:lapl2}
Let $M(x)=g(x)+Z$ with $Z\in\mathbb{R}^d$, $Z_i\sim\mathrm{Lap}(b)$ i.i.d. For any pair $x,x'$ with $\|g(x)-g(x')\|_2\le S$ and any $\delta\in(0,1)$, the pair satisfies Eq.~(\ref{eq:dp}) with
\begin{equation}
\label{eq:lapl2}
\varepsilon \;=\; \frac{S^2}{2b^2}\;+\;\frac{S}{b}\sqrt{2\ln(1/\delta)} ,
\end{equation}
independent of the dimension $d$.
\end{theorem}
\begin{proof}[Proof sketch]
The privacy loss decomposes coordinatewise, $\mathrm{PL}=\frac{1}{b}\sum_i\bigl(|z_i-v_i|-|z_i|\bigr)$ with $z\sim\mathrm{Lap}(b)^d$. Each term is independent and bounded in $[-|v_i|/b,\,|v_i|/b]$ by the triangle inequality, and Lemma~\ref{lem:shift} bounds the mean by $\mathbb{E}[\mathrm{PL}]\le\|v\|_2^2/(2b^2)\le S^2/(2b^2)$. Hoeffding's inequality over the $d$ independent terms, whose squared ranges sum to at most $4S^2/b^2$, gives $\Pr[\mathrm{PL}\ge\mathbb{E}[\mathrm{PL}]+\tau]\le\exp\bigl(-\tau^2b^2/(2S^2)\bigr)$. Choosing $\tau=(S/b)\sqrt{2\ln(1/\delta)}$ and applying Lemma~\ref{lem:tail} completes the proof. Full details are in the supplement.
\end{proof}
The dimension cancels because the fluctuation of the privacy loss is controlled by $\|v\|_2$, the same quantity the spectral projection constrains, rather than by $\|v\|_1$. The guarantee holds directly at every shift size, so certification needs no group privacy.

\begin{corollary}[Calibration and certification]
\label{cor:calib}
Enforce $\Delta_{2,2}\le 1$ at the noise layer. To meet a target $(\varepsilon_0,\delta)$ at construction bound $L$, set
\begin{equation}
\label{eq:calib}
b=\frac{L}{u},\qquad
u=\sqrt{2\ln(1/\delta)+2\varepsilon_0}-\sqrt{2\ln(1/\delta)} .
\end{equation}
At any candidate radius $L'$ the deployed mechanism satisfies $(\varepsilon(L'),\delta)$-DP with $\varepsilon(L')=L'^2/(2b^2)+(L'/b)\sqrt{2\ln(1/\delta)}$, which plugs directly into Proposition~\ref{prop:cert}.
\end{corollary}

\noindent\textbf{Cost and consequences.}
At $(\varepsilon_0,\delta)=(1,0.05)$, Eq.~(\ref{eq:calib}) gives $b=2.64\,L$ and noise standard deviation $\sqrt{2}\,b=3.73\,L$, against $2.54\,L$ for the Gaussian mechanism, a factor of $1.47$. The Gaussian mechanism therefore remains the efficient choice when only $\ell_2$ robustness is wanted. The value of Theorem~\ref{thm:lapl2} is threefold. It closes the certification gap that the repair opened, turning every spectrally constrained Laplace checkpoint into a certified model by re-running only the certification pass. It replaces the naive factor $\sqrt{d}\approx 110.9$ with no dimension dependence at all, a $42\times$ noise reduction at our layer at $\varepsilon_0=1$. And it explains the shape of the original design space, in which the 2-norm Laplace row was absent for exactly this reason.

\subsection{A Gaussian-DP Calibration Variant}
\label{sec:gdp}
Our tables also report a secondary accounting based on Gaussian differential privacy~\cite{dong2022gaussian}. A Gaussian mechanism with $\ell_2$ sensitivity $S$ and scale $\sigma$ is exactly $\mu$-GDP with $\mu=S/\sigma$. The variant fixes a target $\mu_0$ at construction, sets $\sigma=L/\mu_0$, and certifies at radius $L'$ by converting $\mu(L')=L'/\sigma$ to $(\varepsilon,\delta)$ through the duality of~\cite{dong2022gaussian} at the same $\delta$, then applying Proposition~\ref{prop:cert}. Laplace rows under this accounting match the $\varepsilon$ implied by $(\mu_0,\delta)$ at construction. We emphasize that this fixes the privacy level at the construction bound and is an alternative calibration rather than the group-privacy scaling of the original framework. Exact formulas are in the supplement.

\subsection{Certification Meets Adversarial Training}
\label{sec:at}
Proposition~\ref{prop:cert} depends only on the deployed mechanism, not on how the network was trained, so any training procedure that enlarges the score margins under noise raises certified accuracy with the certificate untouched. Following the margin view of adversarially trained smoothing~\cite{salman2019provably}, we train the noisy network with $\ell_2$ projected gradient descent whose inner attack uses expectation-over-transformation gradients~\cite{athalye2018obfuscated}, and we adopt the warmup and weak augmentation guidance of~\cite{debenedetti2023light}. Two protocol choices matter for a fair reading. First, gains must be measured against a matched control that shares every training hyperparameter except the attack, since recipe differences alone move certified accuracy by several points. Second, on class-imbalanced data the joint objective can collapse to the class prior, because the attack suppresses the slow escape of noisy training from the prior basin. A two-phase schedule, noisy pretraining followed by adversarial fine-tuning from the shared checkpoint, removes the failure and yields the cleanest single-variable comparison. Section~\ref{sec:exp_at} reports both phenomena.

\section{Experiments}
\label{sec:experiments}

\subsection{Setup}
\label{sec:setup}
\textbf{Datasets and models.}
We evaluate on CIFAR-10, CIFAR-100~\cite{krizhevsky2009learning}, SVHN~\cite{netzer2011reading}, and ImageNet~\cite{deng2009imagenet}. The ViT experiments use ViT-Tiny with patch size $4$ at resolution $32$, trained from scratch, with noise after the patch embedding. CIFAR models train for $400$ epochs with RandAugment and mixup, SVHN for $300$ epochs with RandAugment only. The convolutional reference is ResNet-50 with a CIFAR stem ($3{\times}3$ first convolution), noise after the first convolution. ImageNet follows the autoencoder configuration of~\cite{lecuyer2019certified}, a three-layer tied-weight autoencoder with noise after its first convolution ($10{\times}10{\times}32$, so $C{=}3200$ coefficients per column), trained first, then frozen while a pretrained backbone is fine-tuned on top for $20$k steps. We run this with both Inception-v3~\cite{szegedy2016rethinking} and ViT-Base backbones. All pipelines consume raw $[0,1]$ pixels without input normalization, as the calibration requires.

\textbf{Mechanisms and accounting.}
Unless stated otherwise, noise is calibrated at $(\varepsilon_0,\delta){=}(1,0.05)$ over the grid $L\in\{0.03,0.1,0.3,1\}$ under $\ell_2$ attack ($\ell_1$ for the $\Delta_{1,1}$ Laplace rows). The $f$-DP variant of Sec.~\ref{sec:gdp} uses $\mu_0{=}0.4$. Sensitivity is renormalized to one after every optimizer step and verified at load time before any certification.

\textbf{Measurement protocol.}
Sweep tables report test accuracy under a single noise draw. Certified tables use $n{=}300$ draws, argmax scores, Clopper--Pearson bounds with a union bound at $\eta{=}0.95$, and report the expected prediction. Sweep tables report test accuracy under a single noise draw. Certified tables use $n=300$ draws, argmax scores, and Clopper--Pearson bounds with a union bound at $\eta=0.95$. These quantities need not coincide, and we do not mix the two protocols within a table.


\begin{table*}[!t]
\centering
\small
\setlength{\tabcolsep}{5pt} 
\caption{Accuracy (\%) of ResNet-50 and ViT under standard DP and
$f$-DP accounting. Gaussian noise uses the $\Delta_{2,2}$ constraint, while Laplace noise uses either the inherited grouped $\bar{\Delta}_{1,1}$ bound or the spectral $\Delta_{2,2}$ constraint. The ViT results correspond
to the 400-epoch runs. Values in parentheses denote the corresponding
non-private baseline accuracies. SVHN chance level is the
majority-class prior of $19.60\%$.}
\label{tab:sweep}
\resizebox{\textwidth}{!}{
\begin{tabular}{lllcccccccc}
\toprule
\multirow{2}{*}{Dataset}
& \multirow{2}{*}{Backbone}
& \multirow{2}{*}{Mechanism}
& \multicolumn{2}{c}{$L{=}0.03$}
& \multicolumn{2}{c}{$L{=}0.10$}
& \multicolumn{2}{c}{$L{=}0.30$}
& \multicolumn{2}{c}{$L{=}1.00$} \\
\cmidrule(lr){4-5}
\cmidrule(lr){6-7}
\cmidrule(lr){8-9}
\cmidrule(lr){10-11}
& & &
DP & $f$-DP &
DP & $f$-DP &
DP & $f$-DP &
DP & $f$-DP \\
\midrule

\multirow{5}{*}{CIFAR-10}
& \multirow{2}{*}{\shortstack[l]{ResNet-50\\(94.48)}}
& Gaussian, $\Delta_{2,2}$
& 91.62 & 91.65
& 82.99 & 83.48
& 68.10 & 68.13
& 45.98 & 46.29 \\
&
& Laplace, $\bar{\Delta}_{1,1}$
& 85.26 & 78.37
& 79.19 & 66.06
& 65.26 & 51.86
& 49.38 & 28.91 \\
\cmidrule(lr){2-11}
&
\multirow{3}{*}{\shortstack[l]{ViT\\(92.92)}}
& Gaussian, $\Delta_{2,2}$
& 84.81 & 84.00
& 70.17 & 70.79
& 42.98 & 42.94
& 21.58 & 22.33 \\
&
& Laplace, $\bar{\Delta}_{1,1}$
& 28.87 & 19.21
& 18.48 & 12.26
& 11.67 & 10.50
& 10.38 & 10.35 \\
&
& Laplace, $\Delta_{2,2}$
& 85.50 & 81.68
& 80.28 & 56.69
& 54.66 & 33.41
& 31.31 & 18.44 \\

\midrule

\multirow{5}{*}{CIFAR-100}
& \multirow{2}{*}{\shortstack[l]{ResNet-50\\(73.28)}}
& Gaussian, $\Delta_{2,2}$
& 68.66 & 68.96
& 57.47 & 57.01
& 39.47 & 40.39
& 22.25 & 22.43 \\
&
& Laplace, $\bar{\Delta}_{1,1}$
& 53.64 & 43.60
& 44.54 & 36.59
& 36.04 & 22.99
& 21.74 & 9.36 \\
\cmidrule(lr){2-11}
&
\multirow{3}{*}{\shortstack[l]{ViT\\(72.79)}}
& Gaussian, $\Delta_{2,2}$
& 56.18 & 55.99
& 39.67 & 40.47
& 19.85 & 19.97
& 5.30  & 5.48 \\
&
& Laplace, $\bar{\Delta}_{1,1}$
& 7.51 & 3.58
& 3.39 & 1.10
& 1.22 & 1.05
& 1.08 & 1.04 \\
&
& Laplace, $\Delta_{2,2}$
& 59.92 & 53.89
& 52.17 & 29.71
& 29.71 & 12.70
& 10.96 & 3.60 \\

\midrule

\multirow{5}{*}{SVHN}
& \multirow{2}{*}{\shortstack[l]{ResNet-50\\(97.12)}}
& Gaussian, $\Delta_{2,2}$
& 96.72 & 96.73
& 94.91 & 95.00
& 85.23 & 85.42
& 50.79 & 51.41 \\
&
& Laplace, $\bar{\Delta}_{1,1}$
& 96.76 & 95.73
& 94.18 & 88.95
& 88.48 & 19.59
& 19.60 & 19.60 \\
\cmidrule(lr){2-11}
&
\multirow{3}{*}{\shortstack[l]{ViT\\(97.01)}}
& Gaussian, $\Delta_{2,2}$
& 95.17 & 95.36
& 83.53 & 81.08
& 19.59 & 41.86
& 19.59 & 19.59 \\
&
& Laplace, $\bar{\Delta}_{1,1}$
& 19.59 & 19.59
& 19.59 & 19.60
& 19.60 & 19.59
& 19.60 & 19.59 \\
&
& Laplace, $\Delta_{2,2}$
& 96.36 & 93.13
& 92.18 & 64.16
& 60.09 & 19.60
& 19.59 & 19.59 \\

\bottomrule
\end{tabular}}
\end{table*}

\subsection{Main Result}
\label{sec:exp_repair}
Table~\ref{tab:sweep} shows the phenomenon and the causal test together. Under the faithful port the $\Delta_{1,1}$ Laplace configuration is destroyed at the smallest noise level, falling to $28.9\%$ on CIFAR-10, $7.5\%$ on CIFAR-100, and exactly the majority-class prior on SVHN at every level. The same Laplace mechanism, the same noise scale, and the same training recipe under the $\Delta_{2,2}$ constraint recover to within a few points of the Gaussian arm, and on CIFAR-10 and CIFAR-100 exceed it at small $L$, consistent with the smaller standard deviation the Laplace mechanism injects. The convolutional reference completes the picture. On ResNet-50, whose first convolution has only $C{=}576$ coefficients per column, the $\Delta_{1,1}$ Laplace arm trains normally on all three datasets (Table~\ref{tab:sweep}), exactly as the norm-ratio law predicts for small $C$. An absorption control with noise after the first attention block instead of the patch embedding leaves SVHN accuracy at the baseline for every $L$ (supplement), confirming that placement, not architecture, determines whether the mechanism binds. 

\begin{figure}[!t]
\centering
\includegraphics[width=0.5\textwidth]{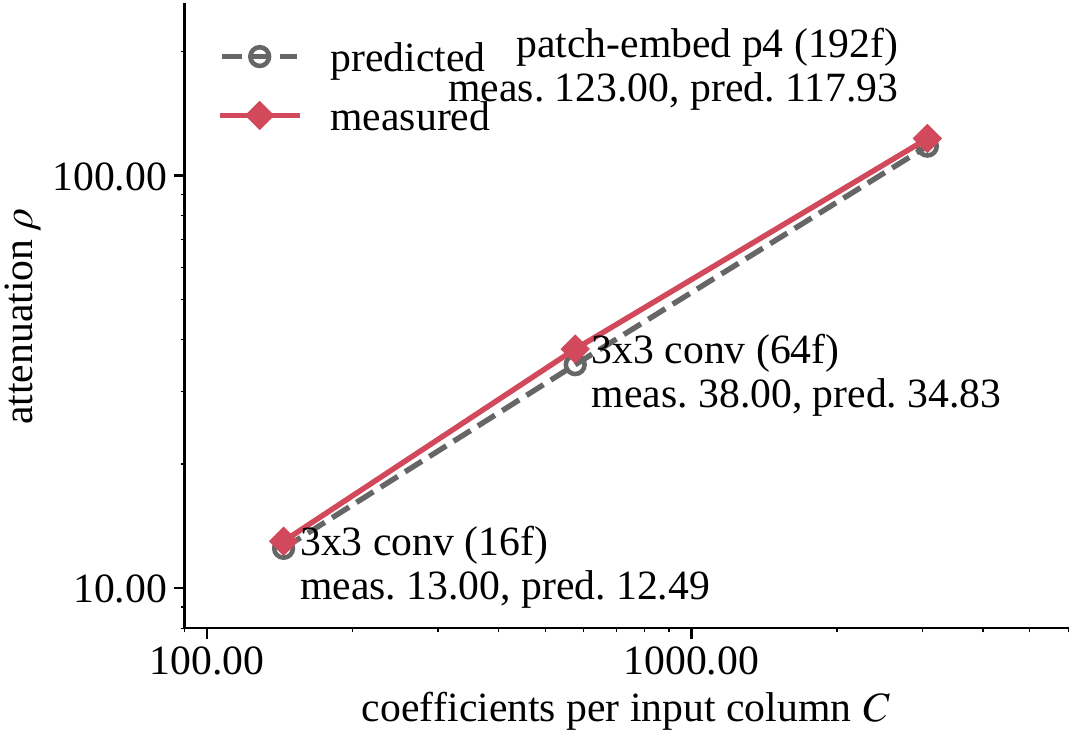}
\caption{The norm-ratio law. Predicted attenuation $C\sqrt{2/\pi}/(\sqrt{M}+\sqrt{N})$ against direct measurement at the three first layers of our study.}
\label{fig:law}
\end{figure}
\par Figure~\ref{fig:law} compares the predicted and measured norm ratios for three first-layer configurations. Using the displayed values, the relative errors with respect to the measurements are approximately 3.9\%, 8.3\%, and 4.1\%. The ImageNet results in Table~\ref{tab:imagenet} provide a separate accuracy-based evaluation of a front end with $C=3200$.
\begin{table}[!t]
\centering
\small
\setlength{\tabcolsep}{1.5pt}
\caption{ImageNet accuracy (\%) under standard DP and $f$-DP
accounting for the autoencoder configuration with two backbones.
Gaussian noise remains relatively stable across the two accounting
schemes, whereas the $\Delta_{1,1}$ Laplace degradation is substantially
more severe under $f$-DP and persists at $C{=}3200$. Values in
parentheses denote the corresponding non-private baseline accuracies.}
\label{tab:imagenet}
\begin{tabular}{llcccc}
\toprule
Backbone & Mechanism
& $L{=}0.03$ & $L{=}0.10$ & $L{=}0.30$ & $L{=}1.00$ \\
\midrule
\multirow{4}{*}{\shortstack[l]{Inception-v3\\(77.81)}}
& DP Gaussian, $\Delta_{2,2}$
& 76.11 & 73.70 & 68.38 & 50.86 \\
& $f$-DP Gaussian, $\Delta_{2,2}$
& 76.30 & 73.55 & 69.34 & 51.00 \\
& DP Laplace, $\bar{\Delta}_{1,1}$
& 58.20 & 36.71 & 13.99 & 2.74 \\
& $f$-DP Laplace, $\bar{\Delta}_{1,1}$
& 40.66 & 14.51 & 3.51 & 0.64 \\
\midrule
\multirow{4}{*}{\shortstack[l]{ViT-Base\\(81.14)}}
& DP Gaussian, $\Delta_{2,2}$
& 79.32 & 76.95 & 71.73 & 51.83 \\
& $f$-DP Gaussian, $\Delta_{2,2}$
& 79.10 & 76.87 & 71.52 & 52.19 \\
& DP Laplace, $\bar{\Delta}_{1,1}$
& 60.50 & 35.92 & 11.72 & 1.89 \\
& $f$-DP Laplace, $\bar{\Delta}_{1,1}$
& 40.38 & 12.13 & 2.47 & 0.55 \\
\bottomrule
\end{tabular}
\end{table}


\subsection{ImageNet at Scale}
Table~\ref{tab:imagenet} reports the faithful ImageNet configuration. The Inception baseline of $77.8\%$ matches the $77.5\%$ of the original paper, validating the pipeline. Two observations follow. First, the accuracy drops of the two backbones are nearly identical at every noise level, for instance $-4.1$ against $-4.2$ points for Gaussian at $L{=}0.1$ and $-27.0$ against $-29.3$ at $L{=}1$. The certificate lives in the noisy front end and the backbone is post-processing, so a stronger frozen backbone shifts the whole curve up without changing its shape, and ViT-Base dominates Inception at every operating point. Second, the $\Delta_{1,1}$ Laplace collapse reappears at scale on the configuration the original paper itself prescribes, whose first convolution holds $3200$ coefficients per column, extending the fan-in series of Fig.~\ref{fig:law} to a fourth architecture and a fourth dataset.

\subsection{Certified Accuracy of the Repaired Laplace Mechanism}
\label{sec:exp_cert}
Corollary~\ref{cor:calib} turns every $\Delta_{2,2}$ Laplace checkpoint of Table~\ref{tab:sweep} into a certified model by re-running only the certification pass with the new $\varepsilon(L')$ map. Table~\ref{tab:lapcert} reports the protocol. The certified columns are produced by the released code and commands in the supplement.
\begin{table}[!t]
\centering\small
\setlength{\tabcolsep}{5pt}
\caption{Certified accuracy (\%) of repaired Laplace
models at the nominal construction setting $L=0.1$.
All columns, including $T=0$, use the Monte Carlo
certification protocol with $n=300$ draws and
Theorem~\ref{thm:lapl2}.
DP/$f$-DP row labels identify the checkpoint
configuration, not different certification theorems.}
\label{tab:lapcert}
\begin{tabular}{lccc}
\toprule
 & $T{=}0$ & $T{=}0.02$ & $T{=}0.03$ \\
\midrule
CIFAR-10 (DP) & 77.19 & 69.12 & 64.34 \\
CIFAR-10 ($f$-DP)  & 55.57 & 51.76 & 50.66 \\
CIFAR-100 (DP) & 55.01 & 48.21 & 46.59 \\
CIFAR-100 ($f$-DP) & 28.22 & 25.20 & 24.48 \\
SVHN (DP)      & 95.02 & 92.47 & 91.60 \\
SVHN ($f$-DP)      & 76.68 & 72.44 & 71.19 \\
\bottomrule
\end{tabular}
\end{table}
The certified radii of the Laplace route are necessarily smaller than the Gaussian ones at matched $(\varepsilon_0,\delta)$, both because the calibrated noise is $1.47\times$ larger and because $\varepsilon(L')$ grows superlinearly, and we report them as the honest price of certifying the heavier-tailed mechanism rather than as a new state of the art. 

\section{Conclusion}
\label{sec:conclusion}
We ported PixelDP to Vision Transformers and found that the port is a stress test the original framework had never received. The placement analysis pins the noise to the patch embedding, the inherited grouped $\ell_1$ Laplace configuration collapses there, stride-aware control recovers much of the lost accuracy, and a closed-form norm-ratio law explains the grouped-bound attenuation across architectures. The repair created a certification gap, and the certified $(\varepsilon,\delta)$ guarantee for the Laplace mechanism under $\ell_2$ sensitivity closes it with a self-contained proof and a closed-form calibration, at an honestly quantified noise cost against the Gaussian mechanism. ImageNet experiments with a fixed autoencoder front end and two downstream backbones illustrate the post-processing property empirically, and a matched-budget ablation shows that adversarial training improves
certified accuracy over the nonzero operating region.

\textbf{Limitations.}
Our ViT models are small and trained from scratch at low resolution, so absolute accuracies sit below modern large-scale recipes, and the study isolates mechanisms rather than chasing state of the art. The Laplace certificate costs about $1.47\times$ the Gaussian noise at matched $(\varepsilon_0,\delta)$, so it removes an inconsistency in the design space rather than improving the efficient frontier. The GDP rows are an alternative calibration that fixes the privacy level at the construction bound and should not be read as the group-privacy scaling of the original framework. Training results are single-seed, so we do not quantify optimization variance. We therefore restrict our empirical claim to the large failure-and-repair effects reproduced across datasets and architectures, rather than small differences between individual configurations. We study $\ell_1$ and $\ell_2$ threats only, and the interaction of the noise layer with $\ell_\infty$ training is untouched.

{
    \small
    \bibliographystyle{ieeenat_fullname}
    \bibliography{main}
}

\appendix
\section{Preliminaries}
\label{sec:prelim}

\subsection{Differential Privacy over a Norm Metric}
Differential privacy bounds how much a randomized computation can change when its input changes a little. PixelDP uses the definition over a general metric~\cite{chatzikokolakis2013broadening}, taking the metric to be an $\ell_p$ norm on images rather than record replacement on databases.

\begin{definition}[$(\varepsilon,\delta)$-DP over an $\ell_p$ metric]
\label{def:dp}
A randomized mechanism $M:\mathbb{R}^n\to\mathcal{Y}$ satisfies $(\varepsilon,\delta)$-DP at scale $L$ in $\ell_p$ if for all $x,x'$ with $\|x-x'\|_p\le L$ and all measurable $T\subseteq\mathcal{Y}$,
\begin{equation}
\label{eq:dp}
\Pr[M(x)\in T]\;\le\; e^{\varepsilon}\,\Pr[M(x')\in T]+\delta .
\end{equation}
\end{definition}
Two properties drive everything that follows. Post-processing preserves the guarantee: any function of $M(x)$ is again $(\varepsilon,\delta)$-DP. Bounded outputs have stable expectations, which is the bridge to robustness.

\begin{lemma}[Expected output stability~\cite{lecuyer2019certified}]
\label{lem:stability}
If $M$ satisfies Def.~\ref{def:dp} and $M(x)\in[0,1]$ almost surely, then for all $\|\alpha\|_p\le L$,
$\mathbb{E}[M(x)]\le e^{\varepsilon}\,\mathbb{E}[M(x+\alpha)]+\delta$.
\end{lemma}

\subsection{The PixelDP Construction}
Write the scoring function of a network as $Q=h\circ g$, where $g$ is the computation before a chosen noise layer and $h$ the computation after it. PixelDP constrains the sensitivity of $g$,
\begin{equation}
\Delta_{p,q}\;=\;\sup_{x\ne x'}\frac{\|g(x)-g(x')\|_q}{\|x-x'\|_p},
\end{equation}
and injects noise calibrated to $\Delta_{p,q}$ and to a construction bound $L$, producing the randomized scores $A(x)=h\bigl(g(x)+Z\bigr)$. Two mechanisms are used. The Laplace mechanism draws $Z_i\sim\mathrm{Lap}(b)$ with $b=\Delta_{p,1}L/\varepsilon$ and yields $(\varepsilon,0)$-DP from an $\ell_1$ sensitivity bound. The Gaussian mechanism draws $Z_i\sim\mathcal{N}(0,\sigma^2)$ with $\sigma=\sqrt{2\ln(1.25/\delta)}\,\Delta_{p,2}L/\varepsilon$ and yields $(\varepsilon,\delta)$-DP from an $\ell_2$ bound, for $\varepsilon\le 1$. By post-processing, $A$ inherits the guarantee of the noised layer, provided no skip connection bypasses the noise.

For the non-overlapping ViT patch embedding (kernel size equal to stride), the exact induced $\ell_1\!\to\!\ell_1$ norm of the flattened patch map is
\begin{equation}
\label{eq:delta11} 
\Delta_{1,1}^{\mathrm{exact}}(W)
=
\max_{j,a,b}\sum_m |W_{m,j,a,b}|.
\end{equation}
The original PixelDP-style convolution treatment instead uses the grouped quantity
\begin{equation}
\label{eq:s1} 
\bar{\Delta}_{1,1}(W)
=
\max_j\sum_{m,a,b}|W_{m,j,a,b}|,
\end{equation}
which is a valid but generally conservative upper bound for this non-overlapping patch operator. Our faithful-port experiments retain this inherited bound; consequently, the collapse diagnosed below should be interpreted as a property of this conservative PixelDP constraint rather than of the exact stride-aware $\ell_1\!\to\!\ell_1$ sensitivity.

\subsection{Certified Prediction}
At inference the expectation $\mathbb{E}[A(x)]$ is estimated by Monte Carlo over $n$ noise draws, with per-class confidence bounds $\hat{\mathbb{E}}^{lb},\hat{\mathbb{E}}^{ub}$ obtained from Hoeffding or Clopper--Pearson intervals~\cite{clopper1934use} and a union bound over classes.

\begin{proposition}[Generalized robustness condition~\cite{lecuyer2019certified}]
\label{prop:cert}
Suppose $A$ satisfies Def.~\ref{def:dp} at scale $L$ in $\ell_p$. If for some class $k$
\begin{equation}
\label{eq:prop2}
\hat{\mathbb{E}}^{lb}\bigl(A_k(x)\bigr)\;>\;e^{2\varepsilon}\max_{i\ne k}\hat{\mathbb{E}}^{ub}\bigl(A_i(x)\bigr)+\bigl(1+e^{\varepsilon}\bigr)\delta,
\end{equation}
then the prediction $\arg\max_k \hat{\mathbb{E}}(A_k(x))$ is robust to every perturbation of $\ell_p$ size at most $L$, with probability at least $\eta$ over the estimation.
\end{proposition}
Given a deployed noise scale, the largest $L$ for which Eq.~(\ref{eq:prop2}) holds is returned as the robustness certificate of the input, searching over $L$ through the map $L\mapsto\varepsilon(L)$ of the mechanism. Certified accuracy at threshold $T$ is the fraction of test points that are correct and certified at radius at least $T$.

\subsection{Relation to Randomized Smoothing}
Randomized smoothing~\cite{cohen2019certified} is the input-noise special case of this construction with Gaussian noise, for which a tighter certificate is available. Placing the noise after a trained layer, which is what makes PixelDP flexible, is also what creates the interaction we study, because the sensitivity constraint then reshapes the layer that carries the signal.

\section{Proofs}
\label{app:proofs}

\subsection{Proof of Proposition~\ref{prop:collapse} (norm-ratio law)}
Let $W$ have i.i.d.\ $\mathcal{N}(0,\tau^2)$ entries. Its flattened patch operator has size $M\times N$ with $N=C_{\rm in}k^2$. For each input channel, the inherited grouped bound aggregates $C=Mk^2$ weights across output channels and spatial kernel locations.

\paragraph{Grouped $\ell_1$ mass.} For any fixed input-channel group, the column mass $f(w)=\sum_{i=1}^{C}|w_i|$ is $\sqrt{C}$-Lipschitz with respect to the Euclidean norm, since
\[
|f(w)-f(w')|
\le \|w-w'\|_1
\le \sqrt{C}\|w-w'\|_2.
\]
For $w=\tau g$ with $g\sim\mathcal N(0,I)$, Gaussian concentration therefore gives
\[
\Pr\!\left[
f(w)<\mathbb E f(w)-t
\right]
\le
\exp\!\left(-\frac{t^2}{2C\tau^2}\right).
\]

\emph{Spectral norm.} By the Davidson--Szarek bound~\cite{vershynin2018high}, for a Gaussian matrix $\Pr[\sigma_{\max}>\tau(\sqrt{M}+\sqrt{N}+t)]\le e^{-t^2/2}$. Setting this to $\delta_2$ gives $\sigma_{\max}\le\tau(\sqrt{M}+\sqrt{N}+\sqrt{2\ln(1/\delta_2)})$ with probability $1-\delta_2$.

\emph{Ratio.} Dividing the two bounds and applying a union bound proves Eq.~(\ref{eq:rho_divice}). The common $\tau$ cancels, so the ratio is scale free, and the leading order is $\rho=\Omega\bigl(C/(\sqrt{M}+\sqrt{N})\bigr)$. For a square-ish matricization $N=\Theta(M)$ this is $\Theta(C/\sqrt{M})=\Theta(\sqrt{C}\,k)$, increasing in both the filter count and the kernel size. \qed

\subsection{Proof of Lemma~\ref{lem:shift} (shifted absolute moment)}
For $z\sim\mathrm{Lap}(0,b)$ with density $\frac{1}{2b}e^{-|z|/b}$ and $v\ge0$ (the case $v<0$ is symmetric),
\[
\mathbb{E}\,|z-v|=\int_{-\infty}^{\infty}\!|z-v|\tfrac{1}{2b}e^{-|z|/b}\,dz .
\]
Splitting at $z=0$ and $z=v$ and integrating each piece in closed form gives $\mathbb{E}\,|z-v|=v+b\,e^{-v/b}$. Since $\mathbb{E}\,|z|=b$,
\[
\mathbb{E}\bigl[|z-v|-|z|\bigr]=v-b+b\,e^{-v/b}=b\,\varphi(v/b),\quad \varphi(t)=t-1+e^{-t}.
\]
Finally $\varphi(t)\le t^2/2$ for all $t\ge0$, because $\psi(t)=t^2/2-\varphi(t)$ satisfies $\psi(0)=0$ and $\psi'(t)=t-1+e^{-t}=\varphi(t)\ge0$. Substituting $t=|v|/b$ gives $\mathbb{E}[|z-v|-|z|]\le v^2/(2b)$. \qed

\subsection{Proof of Lemma~\ref{lem:tail} (tail implies DP)}
Let $S_\varepsilon=\{y:\mathrm{PL}(y)\le\varepsilon\}$. For any event $T$,
\[
\Pr_{M(x)}[T]=\Pr_{M(x)}[T\cap S_\varepsilon]+\Pr_{M(x)}[T\cap S_\varepsilon^c].
\]
On $S_\varepsilon$ the density ratio is at most $e^{\varepsilon}$, so the first term is at most $e^{\varepsilon}\Pr_{M(x')}[T]$. The second term is at most $\Pr_{M(x)}[S_\varepsilon^c]\le\delta$ by hypothesis. Adding gives Eq.~(\ref{eq:dp}). \qed

\subsection{Proof of Theorem~\ref{thm:lapl2}}
Write
\[
v=g(x')-g(x), \qquad \|v\|_2\le S,
\]
and let $y=g(x)+z$ with
$z\sim\operatorname{Lap}(b)^d$. The privacy loss is
\[
\mathrm{PL}=\sum_{i=1}^d\frac{|y_i-g_i(x')|-|y_i-g_i(x)|}{b}
=\frac{1}{b}\sum_{i=1}^d\bigl(|z_i-v_i|-|z_i|\bigr).
\]
They are independent. By the triangle inequality each lies in $[-|v_i|,|v_i|]/b$, so the $i$-th term has range $2|v_i|/b$ and $\sum_i(2|v_i|/b)^2=4\|v\|_2^2/b^2\le 4S^2/b^2$. By Lemma~\ref{lem:shift},
\[
\mathbb{E}[\mathrm{PL}]=\frac{1}{b}\sum_i b\,\varphi(|v_i|/b)
\le\frac{1}{2b^2}\sum_i v_i^2=\frac{\|v\|_2^2}{2b^2}\le\frac{S^2}{2b^2}.
\]
Hoeffding's inequality for bounded independent variables gives, for $\tau>0$,
\[
\Pr[\mathrm{PL}\ge\mathbb{E}[\mathrm{PL}]+\tau]
\le\exp\!\Bigl(\frac{-2\tau^2}{\sum_i(2|v_i|/b)^2}\Bigr)
\le\exp\!\Bigl(\frac{-\tau^2 b^2}{2S^2}\Bigr).
\]
Set $\tau=(S/b)\sqrt{2\ln(1/\delta)}$, so the right side equals $\delta$. Then with probability at least $1-\delta$,
\[
\mathrm{PL}\le\frac{S^2}{2b^2}+\frac{S}{b}\sqrt{2\ln(1/\delta)}=\varepsilon.
\]
Lemma~\ref{lem:tail} converts this tail bound into $(\varepsilon,\delta)$-DP. No step depends on $d$. Because $(x,x')$ was an arbitrary ordered neighboring pair,
the same argument applies after exchanging $x$ and $x'$.
Thus the bound holds for every ordered pair required by
Definition~\ref{def:dp}. \qed

\subsection{Proof of Corollary~\ref{cor:calib} (calibration)}
Enforcing $\Delta_{2,2}\le1$ makes $S=L$ at construction bound $L$. Writing $r=\sqrt{2\ln(1/\delta)}$ and solving $\varepsilon_0=L^2/(2b^2)+(L/b)r$ for $L/b$ as a quadratic gives $L/b=\sqrt{r^2+2\varepsilon_0}-r=u$, hence $b=L/u$. At a candidate radius $L'$ the sensitivity is $S=L'$, and substituting into Eq.~(\ref{eq:lapl2}) yields $\varepsilon(L')=L'^2/(2b^2)+(L'/b)r$. \qed

\section{Gaussian-DP Calibration Formulas}
\label{app:gdp}
A Gaussian mechanism with $\ell_2$ sensitivity $S$ and scale $\sigma$ is $\mu$-GDP with $\mu=S/\sigma$~\cite{dong2022gaussian}. The dual $(\varepsilon,\delta)$ profile is
\[
\delta(\varepsilon;\mu)=\Phi\!\Bigl(-\frac{\varepsilon}{\mu}+\frac{\mu}{2}\Bigr)-e^{\varepsilon}\Phi\!\Bigl(-\frac{\varepsilon}{\mu}-\frac{\mu}{2}\Bigr),
\]
with $\Phi$ the standard normal CDF. Fixing $\mu_0$ at construction sets $\sigma=L/\mu_0$. At radius $L'$ we have $\mu(L')=L'/\sigma=\mu_0 L'/L$, and inverting $\delta(\varepsilon;\mu(L'))=\delta$ for $\varepsilon$ gives the value fed to Proposition~\ref{prop:cert}. This fixes the privacy level at the construction bound rather than scaling it by group privacy.

For the Laplace rows, let $\epsilon_0^{(\mu)}$ denote the solution of
\[
\delta
=
\Phi\!\left(
-\frac{\epsilon_0^{(\mu)}}{\mu_0}+\frac{\mu_0}{2}
\right)
-
e^{\epsilon_0^{(\mu)}}
\Phi\!\left(
-\frac{\epsilon_0^{(\mu)}}{\mu_0}-\frac{\mu_0}{2}
\right).
\]
For the grouped $\ell_1$ Laplace mechanism we set
\[
b=\frac{L}{\epsilon_0^{(\mu)}}.
\]
For the spectrally constrained Laplace mechanism, Corollary~3.2 gives
\[
b=\frac{L}{u_0},
\qquad
u_0=
\sqrt{2\ln(1/\delta)+2\epsilon_0^{(\mu)}}
-
\sqrt{2\ln(1/\delta)}.
\]
At radius $L'$, its privacy parameter is
\[
\epsilon(L')
=
\frac{L'^2}{2b^2}
+
\frac{L'}{b}\sqrt{2\ln(1/\delta)}.
\]

\paragraph{Attention-side absorption control.}
Figure~\ref{fig:svhn-attention-absorption} complements the placement
analysis of Sec.~\ref{sec:placement}. When noise is injected after the first attention
block, test accuracy remains close to the non-private baseline across
the entire construction grid, under both DP and $f$-DP accounting and
for both noise distributions. The absence of systematic degradation as
$L$ increases indicates that training absorbs the perturbation at this
location. This behavior should not be interpreted as enhanced
robustness: because the pre-noise computation crosses an attention block and is not assigned a certified pixel-space sensitivity bound under our protocol, these
models provide no valid PixelDP certificate.
\begin{figure*}[t]
    \centering
    \includegraphics[width=\linewidth]{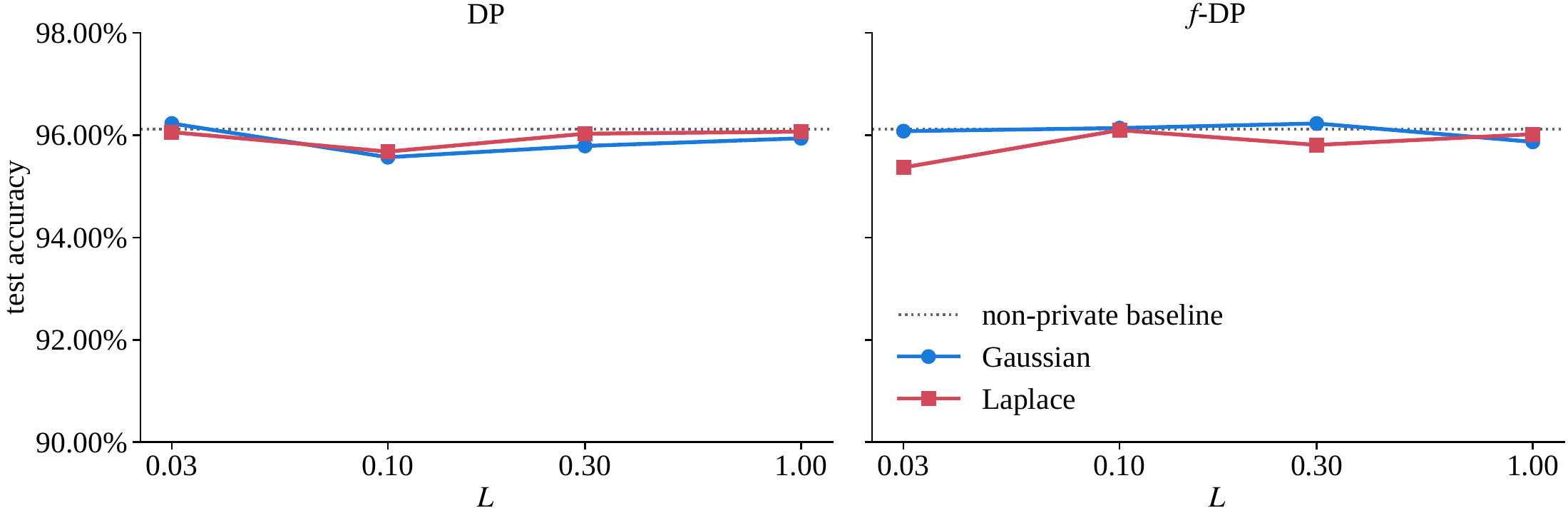}
    \caption{Attention-side absorption control on SVHN. Noise is
    injected after the first attention block rather than after the
    patch embedding. Under both standard DP and $f$-DP accounting,
    Gaussian and Laplace models remain close to the non-private
    baseline across the full construction grid. This apparent
    invariance does not yield a robustness guarantee: no enforceable
    pixel-to-representation sensitivity bound is available at this
    placement, so the resulting models are not certifiable.}
    \label{fig:svhn-attention-absorption}
\end{figure*}

\section{Additional Empirical Studies}
\subsection{The constraint-decoupling Study}
Fig.~\ref{fig:repair} summarizes the constraint-decoupling experiment across CIFAR-10, CIFAR-100, and SVHN. At $L=0.03$, the $\Delta_{1,1}$ Laplace configuration reduces accuracy to 28.87\% on CIFAR-10 and 7.51\% on CIFAR-100, and reaches the majority-class prior on SVHN. Replacing $\Delta_{1,1}$ with $\Delta_{2,2}$ substantially improves low-noise performance. This repair does not eliminate high-noise failures: the SVHN model still reaches the class prior at $L=1$. Exact sweep values and the ResNet-50 reference results are provided in Table~\ref{tab:sweep}.
\begin{figure*}[!t]
\centering
\includegraphics[width=\linewidth]{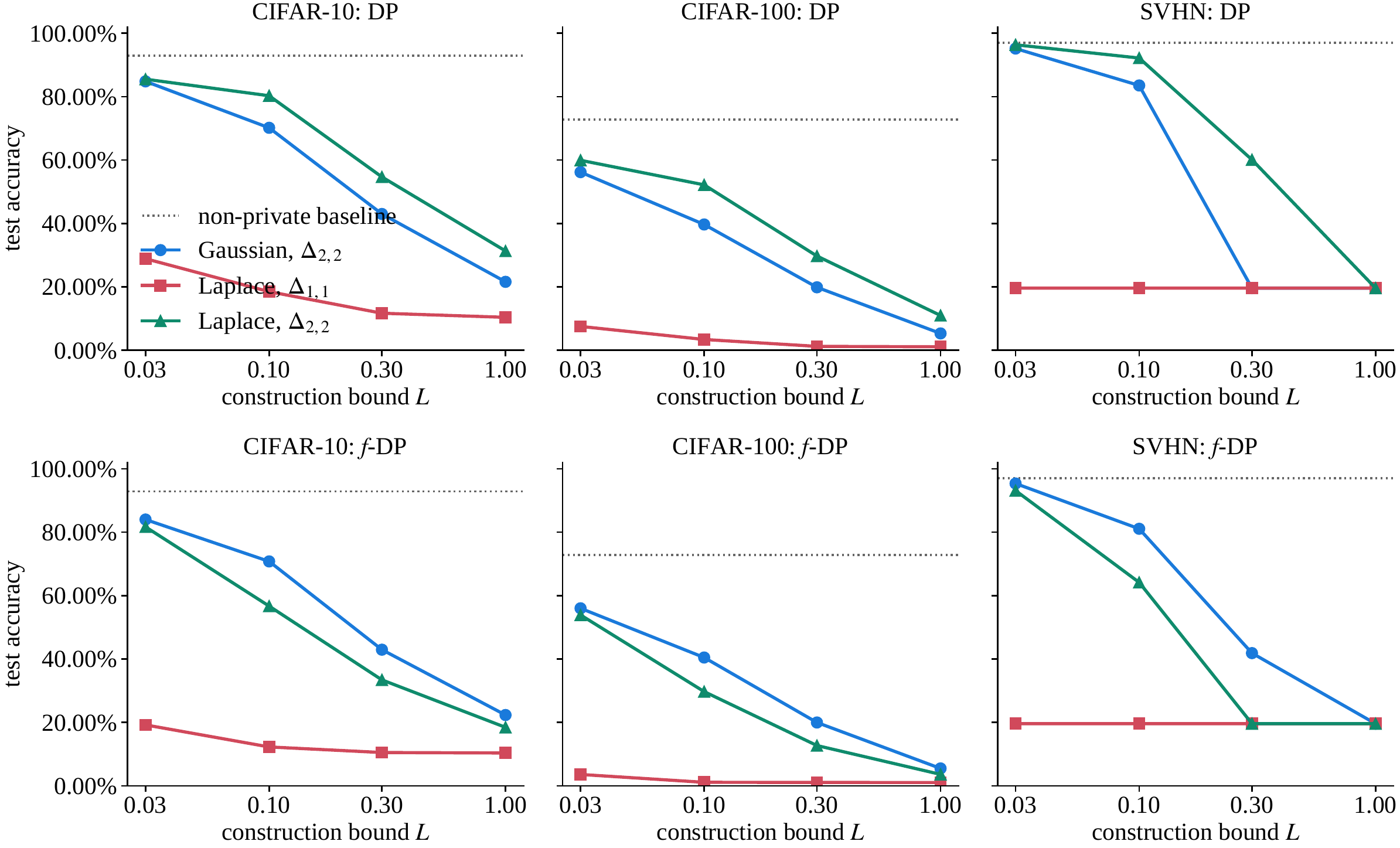}
\caption{Constraint-decoupling diagnostics for ViT-Tiny.
The panels report one-seed test accuracy across
three datasets and two noise-calibration variants.
The spectral constraint substantially mitigates
the low-noise degradation of Laplace models,
although failures remain at large noise levels.
Exact values are provided in
Table~\ref{tab:sweep}.}
\label{fig:repair}
\end{figure*}
\subsection{Sensitivity to Monte Carlo sample size and $\delta$}
Table~\ref{tab:cert_sensitivity} varies the certification sample count and $\delta$ on the same repaired CIFAR-10 checkpoint. Increasing $n$ from 300 to 1,000 changes expected-prediction accuracy by only 0.06 points, but improves certified accuracy by up to 4.65 points over the reported thresholds, consistent with tighter Clopper--Pearson intervals. Reducing $\delta$ from 0.05 to 0.01 has a mixed but predictable effect: it slightly improves the $T=0$ condition, where $\epsilon(0)=0$, but reduces certification at positive radii because $\epsilon(T)$ increases through the $\sqrt{2\ln(1/\delta)}$ term. The main qualitative conclusions are therefore unchanged across these settings.
\begin{table}[!t]
\centering
\caption{Sensitivity of repaired Laplace certification on CIFAR-10 to the number of Monte Carlo samples $n$ and the certification parameter $\delta$. An independently trained $L=0.1$ checkpoint is used for this
sensitivity study; the same checkpoint and the same 1,000 stored
noise draws are used throughout the table.}
\label{tab:cert_sensitivity}
\small
\setlength{\tabcolsep}{4pt}
\begin{tabular}{cc|c|cccc}
\toprule
$n$ & $\delta$ & EP Acc. & $T=0$ & $T=0.02$ & $T=0.03$ & $T=0.05$ \\
\midrule
300  & 0.05 & 81.14 & 77.12 & 69.27 & 64.41 & 51.65 \\
300  & 0.01 & 81.14 & 78.26 & 68.62 & 62.59 & 47.80 \\
1000 & 0.05 & 81.20 & 78.42 & 70.99 & 66.43 & 56.30 \\
1000 & 0.01 & 81.20 & 79.54 & 70.34 & 64.91 & 52.91 \\
\bottomrule
\end{tabular}
\end{table}
\subsection{Norm-ratio evolution during training}
Proposition~\ref{prop:collapse} is derived for i.i.d. Gaussian
weights, so its quantitative prediction is formally an
initialization-scale statement. We therefore track
$\rho=s_1/\sigma_{\max}$ throughout a 400-epoch CIFAR-10 run.
Figure~\ref{fig:rho_training} shows that the Gaussian prediction is
of the correct scale near initialization, but ceases to be
quantitatively tight once optimization induces structure in the
patch embedding: the measured ratio decreases from $136.99$ at
initialization to $53.83$ by epoch 20 and remains near $55$
thereafter, versus the initialization prediction $117.93$.
Importantly, the ratio remains large throughout training, so the
substantial separation between the two constraint geometries
persists even outside the theorem's i.i.d. regime.
\begin{figure*}[!t]
    \centering
    \includegraphics[width=0.8\textwidth]{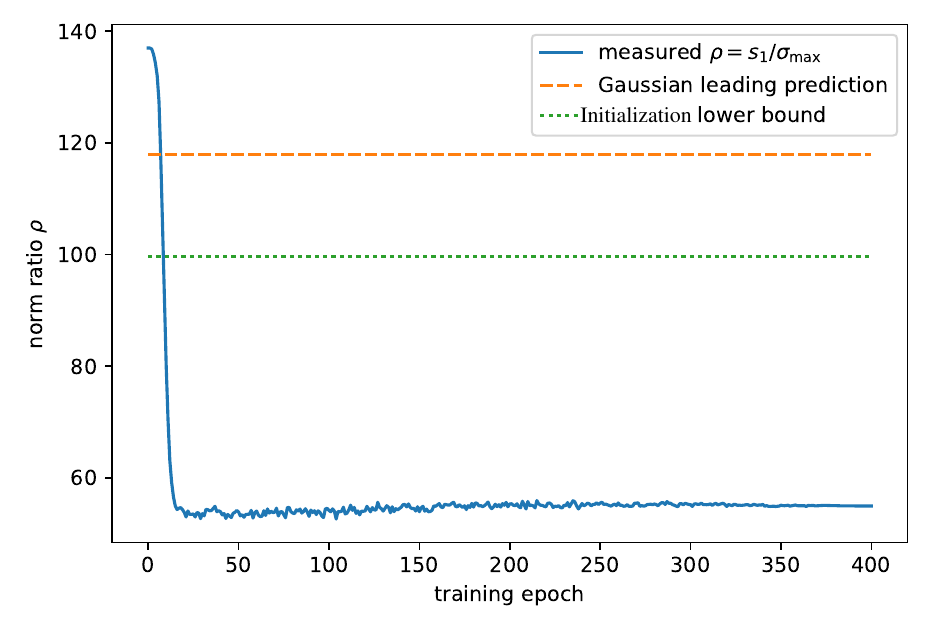}
    \caption{Evolution of the patch-embedding norm ratio
    $\rho=s_1/\sigma_{\max}$ during CIFAR-10 training.
    The dashed line is the leading-order Gaussian prediction, while the dotted line is the initialization lower bound. The prediction is
    initialization-scale rather than a theorem for trained weights:
    optimization rapidly moves the layer outside the i.i.d.\ regime,
    after which $\rho$ stabilizes near $55$. The persistent
    $\rho\gg1$ nevertheless preserves the qualitative separation
    between the inherited $\Delta_{1,1}$ and spectral constraints.}
    \label{fig:rho_training}
\end{figure*}

\paragraph{Exact stride-aware $\ell_1$ control.}
For the non-overlapping ViT patch embedding, we additionally replace the inherited grouped bound $\bar{\Delta}_{1,1}$ with the exact induced $\ell_1\!\to\!\ell_1$ norm $\Delta^{\rm exact}_{1,1}$ of Eq.~(\ref{eq:delta11}), while keeping the Laplace noise scale and training recipe unchanged. At $L=0.03$, accuracy increases from $28.87\%$ to $70.83\%$, and spectral $\Delta_{2,2}$ control further reaches $85.50\%$
(Table~\ref{tab:exact_l1_control}).
\begin{table}[!t]
\centering
\caption{
Constraint-only control on CIFAR-10/ViT-Tiny with Laplace noise at
$L=0.03$. Only the patch-embedding sensitivity constraint is changed.
}
\label{tab:exact_l1_control}
\small
\setlength{\tabcolsep}{6pt}
\begin{tabular}{lc}
\toprule
Constraint & Test Acc. (\%) \\
\midrule
Inherited grouped $\bar{\Delta}_{1,1}$ & 28.87 \\
Exact stride-aware $\Delta^{\rm exact}_{1,1}$ & 70.83 \\
Spectral $\Delta_{2,2}$ & 85.50 \\
\bottomrule
\end{tabular}
\end{table}

On the final exact-norm checkpoint,
$\Delta^{\rm exact}_{1,1}=1.00$,
$\bar{\Delta}_{1,1}=14.14$, and
$\sigma_{\max}=0.290$, giving
$\rho_{\rm exact}=3.45$ and
$\rho_{\rm grouped}=48.78$.
\subsection{Comparison with input randomized smoothing}
\label{app:rs}
On CIFAR-10, we additionally train a standard input Gaussian randomized-smoothing baseline using the same ViT-Tiny architecture and 400-epoch training budget (Table~\ref{tab:rs_baseline}). At $\sigma=0.2537$, its Monte Carlo prediction accuracy (81.57\%) is nearly identical to that of the repaired Laplace model (81.14\%), but its specialized smoothing certificate is substantially tighter: CA@0.05 is 75.75\% versus 51.65\%. This gap is expected and clarifies the scope of our contribution. Our goal is not to outperform input smoothing for pure $\ell_2$ certification; rather, PixelDP permits noise to be placed after an internal representation, and our analysis characterizes the sensitivity-constraint interactions introduced by that flexibility.
\begin{table}[!t]
\centering
\caption{Comparison with standard input Gaussian randomized smoothing
(RS) on CIFAR-10. Both use ViT-Tiny and matched training budgets.
RS uses $\sigma=0.2537$ Gaussian noise at the image input and its
standard smoothing certificate; repaired Laplace injects noise after
the patch embedding and uses the corresponding PixelDP-style
certificate. Pred.\ Acc.\ denotes the Monte Carlo expected prediction.}
\label{tab:rs_baseline}
\footnotesize
\setlength{\tabcolsep}{1pt}
\begin{tabular}{lccccc}
\toprule
Method & Pred.\ Acc. & $T=0$ & $T=0.02$ & $T=0.03$ & $T=0.05$ \\
\midrule
Input Gaussian RS
& 81.57 & 78.72 & 77.63 & 77.01 & 75.75 \\
Repaired Laplace, $\Delta_{2,2}$
& 81.14 & 77.12 & 69.27 & 64.41 & 51.65 \\
\bottomrule
\end{tabular}
\end{table}

\subsection{Adversarial Training Lifts the Certificate}
\label{sec:exp_at}
\begin{table}[!t]
\centering
\small
\setlength{\tabcolsep}{2pt}
\caption{Certified accuracy (\%) on CIFAR-10 at $L{=}0.1$ with
Gaussian noise. The matched noise-only control shares all
hyperparameters with the AT arm except for the adversarial attack.
The strong recipe uses 400 epochs with mixup and is included as a
reference rather than a matched control.}
\label{tab:at-cifar10}
\begin{tabular}{lccccc}
\toprule
Training setting
& $T{=}0.00$ & $T{=}0.05$ & $T{=}0.10$
& $T{=}0.20$ & $T{=}0.30$ \\
\midrule
Matched noise only
& 64.65 & 45.04 & 32.67 & 0.00 & 0.00 \\
AT (PixelDP)
& 66.62 & 53.34 & 43.37 & 0.00 & 0.00 \\
AT ($f$-DP)
& 66.42 & 53.28 & 43.02 & 0.00 & 0.00 \\
Strong recipe
& 74.30 & 57.66 & 45.27 & 0.00 & 0.00 \\
\bottomrule
\end{tabular}
\end{table}

\begin{table}[!t]
\centering
\tiny
\setlength{\tabcolsep}{1pt}
\caption{Certified accuracy (\%) on CIFAR-100 and SVHN at
$L{=}0.1$ with Gaussian noise. For CIFAR-100, the matched noise-only
control shares all hyperparameters with the AT arm except for the
adversarial attack, while the strong recipe is included as a reference.
For SVHN, the two main branches share a 300-epoch noisy prefix,
followed by either 150 additional noise-only epochs or 150 epochs of
AT fine-tuning.}
\label{tab:at-cifar100-svhn}
\begin{tabular}{llccccc}
\toprule
Dataset & Training setting
& $T{=}0.00$ & $T{=}0.05$ & $T{=}0.10$
& $T{=}0.125$ & $T{=}0.15$ \\
\midrule
\multirow{3}{*}{CIFAR-100}
& Matched noise only
& 31.28 & 18.40 & 12.43 & 9.40  & 5.18 \\
& AT (PixelDP)
& 30.20 & 21.08 & 16.08 & 13.47 & 9.40 \\
& AT ($f$-DP)
& 32.74 & 22.51 & 16.80 & 13.98 & 9.85 \\
& Strong recipe
& 42.96 & 27.63 & 19.93 & 15.80 & 9.98 \\
\midrule
\multirow{3}{*}{SVHN}
& Noise only (PixelDP), 450 epochs
& 83.48 & 63.74 & 46.34 & 36.41 & 27.42 \\
& Noise only ($f$-DP), 450 epochs
& 84.52 & 65.52 & 49.10 & 39.61 & 30.76 \\
& AT fine-tuning (PixelDP), 150 epochs
& 82.90 & 66.55 & 51.32 & 42.40 & 33.51 \\
& AT fine-tuning ($f$-DP), 150 epochs
& 83.42 & 66.79 & 52.48 & 44.10 & 35.56 \\
& Noise only (PixelDP), 300 epochs
& 82.63 & 62.02 & 44.29 & 34.80 & 25.88 \\
& Noise only ($f$-DP), 300 epochs
& 86.08 & 66.38 & 48.74 & 38.47 & 29.12 \\
\bottomrule
\end{tabular}
\end{table}

\begin{figure*}[!t]
\centering
\includegraphics[width=\textwidth]{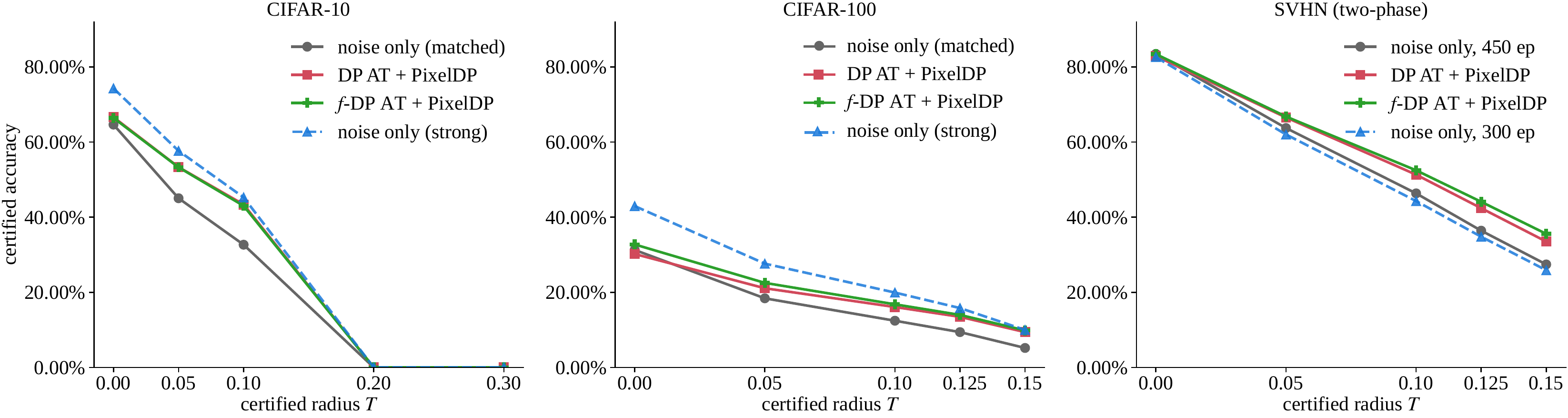}
\caption{Certified accuracy against radius under the matched protocol. Adversarial training improves certified accuracy across the
reported nonzero operating region under the matched protocol.}
\label{fig:at}
\end{figure*}

Table~\ref{tab:at-cifar10}, Table~\ref{tab:at-cifar100-svhn}, and Fig.~\ref{fig:at} give the ablation. Under standard DP calibration, adversarial training improves CIFAR-10 certified accuracy by 8.30 and 10.70 percentage points at $T=0.05$ and $T=0.10$. Across the reported positive thresholds up to $T=0.15$, the gains increase from 2.68 to 4.22 points on CIFAR-100 and from 2.81 to 6.09 points on SVHN. Both CIFAR-10 models have zero certified accuracy at $T=0.20$ and $T=0.30$. The pattern is the margin mechanism of~\cite{salman2019provably} observed through a certificate. The attack enlarges the score gap under noise, and the higher the certification threshold, the more of that gap is converted into certified points. Near the certification boundary the matched AT model approaches the strong recipe trained with $2.7\times$ the budget, on CIFAR-100 within $0.6$ points at $T{=}0.15$. The comparison against the strong recipe also documents a pitfall. Read against that reference alone, adversarial training appears to lose everywhere, and only the matched control reveals that the deficit is the training budget, not the attack.

\textbf{Two-phase training on SVHN.}
Training the joint objective from scratch on SVHN collapses to the class prior of $19.6\%$ and stays there for $300$ epochs. The noisy gradient signal through the constrained embedding is weak early in training, the class imbalance offers a genuine prior basin, and the inner attack suppresses the slow escape that noise-only training exhibits. Forking a $300$-epoch noisy checkpoint into an adversarial fine-tune and a continued noise-only control isolates the attack cleanly and restores the pattern. The collapsed run is also instructive in itself. A constant classifier certifies every one of its majority-class predictions at every radius, so certified accuracy in isolation can be gamed by a degenerate model, and we recommend always reading it jointly with conventional accuracy.

\begin{table}[!t]
\centering\small
\setlength{\tabcolsep}{5pt}
\caption{Empirical $\ell_2$ robustness (\%) on CIFAR-10 at $\varepsilon{=}0.5$, evaluated using the best checkpoints. $\dagger$ denotes evaluation of the averaged randomized prediction
rule. ``--'' indicates that the AT-only model has no valid PixelDP certificate.}
\label{tab:attacks}
\begin{tabular}{lcc}
\toprule
Attack & AT only & AT + PixelDP \\
\midrule
Clean          & 75.84 & 66.72 \\
FGSM           & 55.14 & 54.95 \\
PGD-20         & 54.15 & 51.10 \\
PGD-100        & 54.10 & 50.54 \\
CW-100         & 53.39 & 49.80 \\
BPDA           & 54.19 & 55.26 \\
EOT-PGD        & 54.17 & 50.21 \\
AutoAttack     & 52.26 & 41.52$^\dagger$ \\
\midrule
Certified $T{=}0.10$ & -- & 43.37 \\
\bottomrule
\end{tabular}
\end{table}

\textbf{Empirical robustness and sanity checks.}
Table~\ref{tab:attacks} places the certified arm next to a pure adversarial-training baseline under identical attacks. The AT-only model is empirically stronger, but it has no PixelDP certificate because it lacks the required randomized, sensitivity-controlled mechanism. The table also passes the standard adaptive-evaluation checks. For the deterministic model, BPDA coincides with PGD-20 as it must. For the randomized model, EOT is stronger than PGD and BPDA is weaker, the expected ordering when gradients are averaged over noise~\cite{athalye2018obfuscated}.

\subsection{Takeaways from the Experiment}
\begin{itemize}
\item Placement is binding. Only the patch embedding supports an enforceable sensitivity bound in a ViT, and attention-side noise is absorbed into a certificate-free model.
\item The inherited grouped constraint causes severe attenuation, whose initialization scale follows a closed-form law in the grouped coefficient count from $C=144$ to $C=3200$.
\item An $\ell_2$ certificate for the Laplace mechanism exists, is dimension free, and re-certifies repaired checkpoints without retraining, at a quantified $1.47\times$ noise cost against Gaussian.
\item The certificate lives in the front end. Swapping a frozen Inception for a frozen ViT-Base moves ImageNet baselines, not degradation curves.
\item Adversarial training improves certified accuracy over the nonzero
operating region, but only a matched-budget control isolates the gain.
\end{itemize}
\section{Reproduction}
\label{app:repro}
\textbf{Hyperparameters.} ViT-Tiny, patch $4$, resolution $32$, AdamW, weight decay $0.1$ (AT) or $0.05$ (clean), cosine schedule with linear warmup. CIFAR $400$ epochs with RandAugment and mixup, SVHN $300$ epochs with RandAugment. Certification uses $n{=}300$ draws, argmax scores, Clopper--Pearson at $\eta{=}0.95$. Adversarial training uses $\ell_2$ PGD with $\epsilon_{\rm adv}=0.5$, $K=7$ steps, and step size $\alpha=2.5\epsilon_{\rm adv}/K$. For PixelDP models, each PGD step uses EOT with two independent noise draws and averages the gradients of the per-draw cross-entropy losses. The attack budget is linearly warmed up over the first 30 epochs.

\textbf{Certifying repaired Laplace models.} With the released code, a $\Delta_{2,2}$ Laplace checkpoint is certified under Theorem~\ref{thm:lapl2} by the certification pass with \texttt{--noise\_mech laplace\_l2}, which uses the calibration $b{=}L/u$ of Eq.~(\ref{eq:calib}) and the map $\varepsilon(L')$ of Corollary~\ref{cor:calib}. The load-time sensitivity guard refuses any checkpoint whose $\Delta_{2,2}\ne1$, so an accidentally supplied AT-only checkpoint cannot produce an invalid certificate.

\textbf{Two-phase SVHN.} Train a noise-only checkpoint for $300$ epochs, then fork into an adversarial fine-tune and a continued noise-only control, both for $150$ epochs from the shared weights with reset schedules, so the two forks differ only in the presence of the attack.

\textbf{AutoAttack protocol.} For randomized checkpoints AutoAttack attacks and judges the averaged prediction over $8$ draws, and accuracy is reported over $20$ draws. Judging AutoAttack per single draw, as an off-the-shelf run does, underestimates the deployed averaged prediction, so we report the averaged-rule number and mark it accordingly.


\end{document}